\documentclass[11pt,a4paper]{article}          

\usepackage{amsthm}
\usepackage{amsmath}
\usepackage{amssymb}
\usepackage{latexsym}

\usepackage{graphicx}
\usepackage{framed}
\usepackage[disableredefinitions]{complexity}

\usepackage{tikz}
\usetikzlibrary{decorations.pathreplacing,shapes}
\usepackage{enumitem}
\usepackage{fullpage}

\newtheorem{thm}{Theorem}[section]
\newtheorem{lma}[thm]{Lemma}
\newtheorem{cor}[thm]{Corollary}

\newtheorem{claim}[thm]{Claim}

\newtheorem{fact}[thm]{Fact}

\DeclareMathOperator{\vol}{vol}

\newcommand{\cA}{\mathcal{A}}
\renewcommand{\AA}{\mathcal{A}}

\newcommand{\BB}{\mathcal{B}}

\newcommand{\eps}{\varepsilon}
\newcommand{\ds}{{\rm degsum}}

\newcommand{\tw}{{\rm tw}}

\newcommand{\pw}{{\rm pw}}
\newcommand{\fvs}{{\rm fvs}}

\newcommand{\q}{q^*}

\newcommand{\Fc}{  }

\begin{document}

\title{The parameterised complexity of computing the maximum modularity of a graph\thanks{An extended abstract of this appeared in proc. IPEC 2018.}}

\author{
   Kitty Meeks\thanks{KM is supported by a Royal Society of Edinburgh Personal Research Fellowship, funded by the Scottish Government.}\\
   School of Computing Science\\
   University of Glasgow\\
   Glasgow, UK\\
   \texttt{kitty.meeks@glasgow.ac.uk}
   \and Fiona Skerman\thanks{FS conducted part of this research while at Uppsala University.}\\
   Heilbronn Institute for Mathematical Research\\
   Department of Mathematics\\
   University of Bristol\\
   Bristol, UK\\
   \texttt{f.skerman@bristol.ac.uk}
}

\date{October 2019}

\maketitle

\begin{abstract}
The \emph{maximum modularity} of a graph is a parameter widely used to describe the level of clustering or community structure in a network.  Determining the maximum modularity of a graph is known to be $\NP$-complete in general, and in practice a range of heuristics are used to construct partitions of the vertex-set which give lower bounds on the maximum modularity but without any guarantee on how close these bounds are to the true maximum.  In this paper we investigate the parameterised complexity of determining the maximum modularity with respect to various standard structural parameterisations of the input graph~$G$.  We show that the problem belongs to $\FPT$ when parameterised by the size of a minimum vertex cover for~$G$, and is solvable in polynomial time whenever the treewidth or max leaf number of~$G$ is bounded by some fixed constant; we also obtain an FPT algorithm, parameterised by treewidth, to compute any constant-factor approximation to the maximum modularity.  On the other hand we show that the problem is W[1]-hard (and hence unlikely to admit an FPT algorithm) when parameterised simultaneously by pathwidth and the size of a minimum feedback vertex set.
\end{abstract}

\section{Introduction}

The increasing availability of large network datasets has led to great interest in techniques to discover network structure. An important and frequently observed structure in networks is the existence of groups of vertices with many connections between them, often referred to as `communities'.

Newman and Girvan introduced the modularity function in 2004~\cite{NewmanGirvan}. Modularity gives a measure of how well a graph can be partitioned into communities and is used in the most popular algorithms to cluster large networks.  For example, the Louvain method, an iterative clustering technique, uses the modularity function to choose which parts from the previous step to fuse into larger parts at each step~\cite{jeubcode,louvain}. The widespread use of modularity and empirical success in finding communities makes modularity an important function to study from an algorithmic point of view.

In this paper we are concerned with the computational complexity of computing the maximum modularity of a given input graph, and specifically in the following decision problem.

\begin{framed}
\noindent \textsc{Modularity}\\
\textit{Input:} A graph~$G$ and a constant $q \in [0,1]$.\\
\textit{Question:} Is the maximum modularity of~$G$ at least $q$?
\end{framed}

This problem was shown to be $\NP$-complete in general by Brandes~et.~al.~\cite{brandes08}, using a construction that relies on the fact that all vertices of a sufficiently large clique must be assigned to the same part of an optimal partition.  They also showed that a variation of the problem in which we wish to find the optimal partition into exactly two sets is hard; their proof for this relied again on the use of large cliques, but DasGupta and Desai \cite{dasgupta13} later showed that this 2-clustering problem remains $\NP$-complete on $d$-regular graphs for any fixed~$d \geq 9$.  It has also been shown that it is $\NP$-hard to approximate the maximum modularity within any constant factor \cite{dinh15}, although there is a polynomial-time constant-factor approximation algorithm for certain families of scale-free networks \cite{dinh13}.  The hardness of computing constant-factor multiplicative approximations in general has motivated research into approximation algorithms with an additive error \cite{dinh15,kawase16}: the best known result is an approximation algorithm with additive error roughly 0.42084 \cite{kawase16}.

In this paper we initiate the study of the parameterised complexity of \textsc{Modularity}, considering its complexity with respect to several standard structural parameterisations.  On the positive side, we show that the problem is in $\FPT$ when parameterised by the cardinality of a minimum vertex cover for the input graph~$G$, and that it belongs to $\XP$ when parameterised by either the treewidth or max leaf number of~$G$.  The XP algorithm parameterised by treewidth can easily be adapted to give an FPT algorithm, parameterised by treewidth, to compute any constant-factor approximation maximum modularity.  On the other hand, we demonstrate that \textsc{Modularity}, parameterised by treewidth, is unlikely to belong to $\FPT$: we prove that the problem is $\W[1]$-hard even when parameterised simultaneously by the pathwidth of~$G$ and the size of a minimum feedback vertex set for~$G$.  For background on parameterised complexity, and the complexity classes discussed here, we refer the reader to \cite{paramalgs,downeyfellows13}.

These results follow the same pattern as those obtained for the problem \textsc{Equitable Connected Partition} \cite{enciso09}, and indeed our hardness result involves a reduction from a specialisation of this problem.  There are clear similarities between the two problems: in a partition that maximises the modularity, every part will induce a connected subgraph and, in certain circumstances, we achieve the maximum modularity with a partition into parts that are as equal as possible.  However, the crucial difference between the two problems is that the input to \textsc{Equitable Connected Partition} includes the required number of parts, whereas \textsc{Modularity} requires us to maximise over all possible partition sizes; in fact, if we restrict to partitions with a specified parts, it is no longer necessarily true that a partition maximising the modularity must induce connected subgraphs.  This difference makes reductions between the two problems non-trivial.

\subsection{The modularity function}
\label{sec:mod-fn}
The definition of modularity was first introduced by Newman and Girvan in~\cite{NewmanGirvan}. Many or indeed most popular algorithms used to search for clusterings on large datasets are based on finding partitions with high modularity~\cite{popular,fortunato2016community}, and the heuristics within them sometimes also use local modularity optimisation, for example in the Louvain method~\cite{louvain}. See~\cite{fortunato2010community,porter2009communities} for surveys on community detection including modularity based methods.

Knowledge on the maximum modularity for classes of graphs helps to understand the behaviour of the modularity function. There is a growing literature on this which began with cycles and complete graphs in~\cite{brandes08}. Bagrow~\cite{bagrow} and Montgolfier et al.~\cite{modgraphclasses} showed some classes of trees have high maximum modularity which was extended in~\cite{treelike} to all trees with maximum degree~$o(n)$, and furthermore to all graphs where the product of treewidth and maximum degree grows more slowly than the number of edges. Many random graph models also have high modularity, see~\cite{modERAofA,ERus} for a treatment of Erd\H{o}s-Renyi random graphs, \cite{treelike}~for random regular graphs and also~\cite{pralat} which includes the preferential attachment model.

Given a set~$A$ of vertices, let $e(A)$ denote the number of edges within $A$, and let $\vol(A)$ (sometimes called the volume of $A$) denote the sum of the degree $d_v$ (in the whole graph~$G$) over the vertices $v$ in $A$. For a graph~$G$ with $m\geq 1$ edges and a vertex partition $\cA$ of~$G$, 
set the modularity score of $\cA$ on~$G$ to be 
\[ q_\cA(G) = 
\frac{1}{2m}\sum_{\A\in \cA} \sum_{u,v \in A} 
\left( {\mathbf 1}_{uv\in E} - \frac{d_u d_v}{2m} \right)
= \frac{1}{m}\sum_{\A \in \cA} e(\A) - \frac{1}{4m^2}\sum_{\A\in \cA} \vol (\A)^2; \]
the maximum modularity of~$G$ is $\q(G)=\max_\cA(G)$, where the maximum is over all partitions $\cA$ of the vertices of~$G$. Graphs with no edges are defined conventionally to have modularity 1. However note that if the modularity of graphs with no edges were defined to be 0 it would not change any of the results.

The modularity function is designed to score partitions highly when most edges fall within the parts and penalise partitions with very few or very big parts. These two objectives are encoded as the \emph{edge contribution} or \emph{coverage}  $q^E_\cA(G)=\frac{1}{m}\sum_{\A \in \cA} e(\A)$, and \emph{degree tax} $q_\cA^D(G)=\frac{1}{4m^2}\sum_{\A\in \cA} \vol (\A)^2$,  in the modularity of a vertex partition~$\cA$ of~$G$.

Note that for any graph with $m\geq 1$ edges $0 \leq q^*(G) \leq 1$.  To see the lower bound, notice that the trivial partition which places all vertices in the same part has modularity zero. For example, complete graphs and stars have modularity~0 as noted in
\cite{brandes08}.  A graph consisting of $c$ disjoint cliques of the same size has modularity $1-1/c$ with the optimal partition taking each clique to be a part.

As modularity is at most~1 it is sometimes useful to consider the \emph{modularity {\Fc deficit}} $\tilde{q}_\cA(G)=1-q_\cA(G)$. Denote by $\partial(A)$ the number of edges between vertex set $A$ and the rest of the graph. Then 
\[\tilde{q}_\cA(G) =\frac{1}{2m}\sum_{A\in \cA} \bigg(\partial(A)+\frac{\vol(A)^2}{2m}\bigg) \] and we may equivalently minimise the modularity {\Fc deficit} to maximise the modularity. In particular \[\tilde{q}(G)= \min_{A\in \cA} \tilde{q}_\cA(G) = 1-\q(G).\]

We will make use of several facts about the maximum modularity of a graph.

\begin{fact}[Lemma~1 of \cite{dinh2011finding}, Lemma~2.1 of \cite{dasgupta13}]\label{fact:approxbycALT}
For any integer $c>0$ and any graph~$G$,
$$\max_{|\AA|\leq c} q_\AA(G)> \q(G)\Big(1-\frac{1}{c}\Big).$$
\end{fact}

\begin{fact}[Lemma~3.4 of~\cite{brandes08}]\label{fact:connected}
Suppose that~$G$ is a graph that contains no isolated vertices.  If $\mathcal{A}$ is a partition of $V(G)$ such that $q_{\mathcal{A}}(G) = q^*(G)$ then, for every $A \in \mathcal{A}$, $G[A]$ is a connected subgraph of~$G$.
\end{fact}

\begin{fact}[Corollary~1 of~\cite{brandes08}]\label{fact:isolated}
Let $G = (V,E)$ and suppose that $V_0 \subseteq V$ is a set of isolated vertices.  Then $q(G) = q(G \setminus V_0)$.  Moreover, if partitions $\cA$ and $\cA'$ agree on all vertices of $V \setminus V_0$, then $q_{\cA}(G) = q_{\cA'}(G)$.
\end{fact}
{\Fc
\begin{fact}[Lemma~1.6.5 of~\cite{thesis}]\label{fact:singlevertex}  
If $\mathcal{A}$ is a partition of $V(G)$ such that $q_{\mathcal{A}}(G) = q^*(G)$ then no part $A$ consists of a single non-isolated vertex.
\end{fact}

%\begin{fact}[Lemma~3.3 of~\cite{brandes08}]\label{fact:pendant}  
%If $\mathcal{A}$ is a partition of $V(G)$ such that $q_{\mathcal{A}}(G) = q^*(G)$ then no part $A$ consists of a single vertex of degree~1.
%\end{fact}

%%%%%%%%%%%%%%%%%%%%%%%%%%%%%%%%%%%
%The last result of Brandes et al.\ which we will survey says that any node with degree~1 will not be alone in any optimal partition of our graph (Lemma~3.3 of~\cite{nphard}). We extend this statement to any node of non-zero degree. \\

%\begin{lma}[Brandes et al.~\cite{nphard}]\label{lem.nodegone} Let~$G$ be a graph and $\AA$ an optimal vertex partition of~$G$. Then $A=\{u\}$ for some $A\in \AA$ implies deg$(u)\neq1$. \end{lma}

%A bit more is true. If a single vertex is a part in an optimal partition of~$G$ then it must have degree of~0.\\

%\begin{lma}\label{lem.noPartOne} Let~$G$ be a graph and $\AA$ an optimal vertex partition of~$G$. Then $A=\{u\}$ for some $A \in \AA$ implies deg$(u)=0$.\end{lma}
%%%%%%%%%%%%%%%%%%%%%%%%%%%%%%%%%%%

\begin{proof}
Let $u$ be a vertex with degree $d_u>0$ and suppose (for a contradiction) that $\cA=\{\{u\},A_1, \ldots, A_k\}$ is an optimal partition of~$G$. For each $i=1,\ldots, k$ %write $d_\ell=e(\{u\}, A_\ell)$ and $w_\ell=\ds(A_\ell)$ and 
define the vertex partition $\BB_i=\{A_1, \ldots, A_i\cup\{u\}, \ldots, A_k\}$. We can derive a simple expression for $q_{\BB_i}(G)-q_{\AA}(G)$ as most terms cancel 
\begin{equation*}
q_{\BB_i}(G)-q_{\AA}(G)= \frac{1}{m}e(\{u\},A_i) - \frac{1}{2m^2}d_u\ds(A_i)%=\frac{d_i}{m}-\frac{dw_i}{2m^2}.
\end{equation*}
By assumption, $\AA$ is an optimal partition so $q_{\BB_i}(G)\leq q_\AA(G)$ and thus for each~$i$ we have $2m  \cdot e(\{u\},A_i) \leq d_u \ds(A_i)$. Hence we can sum over $i=1, \ldots, k$ and the inequality should hold. However for the LHS $2m \sum_i e(\{u\},A_i) =2m d_u$ and the RHS is
$$d_u \sum_{i=1}^k  \ds(A_i) = d_u(2m-d_u ) < 2md_u $$ and so we have our contradiction.\end{proof}
}

{\Fc
Observe that Facts \ref{fact:connected}, \ref{fact:isolated} and \ref{fact:singlevertex} together imply that the search for an optimal partition can be restricted to those in which all parts are connected subgraphs and no part consists of a single node.\\
}

\subsection{Notation and definitions}
\label{sec:notation}

Given a graph $G = (V,E)$, and a set $U \subseteq V$ of vertices, we write $G[U]$ for the subgraph of~$G$ induced by $U$ and $G \setminus U$ for $G[V \setminus U]$.  Given two disjoint subsets of vertices $A,B \subseteq V$, we write $e(A,B)$ for the number of edges with one endpoint in $A$ and the other in $B$.  We shall often want to denote the number of edges between a set of vertices and the remainder of the graph so set $\partial(A)=e(A,\bar{A})$.  If $\mathcal{P}$ is a partition of a set $X$, and $Y \subset X$, we write $\mathcal{P}[Y]$ for the restriction of $\mathcal{P}$ to $Y$.

A \emph{vertex cover} of a graph $G = (V,E)$ is a set $U \subseteq V$ such that every edge has at least one endpoint in $U$; equivalently, $G \setminus U$ is an independent set (i.e.\ contains no edges).  The \emph{vertex cover number} of~$G$ is the smallest cardinality of any vertex cover of~$G$.  A \emph{feedback vertex set} for~$G$ is a set $U \subseteq V$ such that $G \setminus U$ contains no cycles.  Notice that the vertex cover number of~$G$ gives an upper bound on the size of the smallest feedback vertex set for~$G$, written $\fvs(G)$.  The \emph{max leaf number} of~$G$ is the maximum number of leaves (degree one vertices) in any spanning tree of~$G$.  

A \emph{tree decomposition} of a graph~$G$ is a pair $(T,\mathcal{D})$ where $T$ is a tree and $\mathcal{D} = \{\mathcal{D}(t): t \in V(T)\}$ is a collection of non-empty subsets of $V(G)$ (or \emph{bags}), indexed by the nodes of $T$, satisfying:
\begin{enumerate}
\item $V(G) = \bigcup_{t \in V(T)} \mathcal{D}(t)$,
\item for every $e=uv \in E(G)$, there exists $t \in V(T)$ such that $u,v \in \mathcal{D}(t)$,
\item for every $v \in V(G)$, if $T(v)$ is defined to be the subgraph of $T$ induced by nodes $t$ with $v \in \mathcal{D}(t)$, then $T(v)$ is connected.
\end{enumerate}
We will assume throughout that the indexing tree $T$ has a distinguished root node $r$; if not we may choose an arbitrary node to be the root.  Given any node $t \in V(T)$ we write $V_t$ for the set of vertices of~$G$ that appear in bags indexed by $t$ and the descendants of $t$.

If $T$ is in fact a path, we say that $(T,\mathcal{D})$ is a \emph{path decomposition} of~$G$.  The \emph{width} of the tree decomposition $(T,\mathcal{D})$ is defined to be $\max_{t \in V(T)} |\mathcal{D}(t)| - 1$, and the \emph{treewidth} of~$G$, written $\tw(G)$, is the minimum width over all tree decompositions of~$G$.  The \emph{pathwidth} of~$G$, $\pw(G)$, is the minimum width over all path decompositions of~$G$.

We note that there is an FPT algorithm to compute a minimum-width tree decomposition of any graph~$G$, where the treewidth of~$G$ is taken as the parameter \cite{bodlaender93}.  Moreover, any such tree decomposition can be transformed into a so-called \emph{nice} tree decomposition (having certain algorithmically useful properties) in linear time, without increasing the number of nodes by more than a constant factor \cite{kloks94}.

\section{Positive results}
\label{sec:positive}

In this section we identify a number of structural restrictions on the input graph that allow us to compute the maximum modularity of a graph, or a good approximation to this quantity, efficiently.

\subsection{Parameterisation by vertex cover number}
\label{sec:vc}

In this section we demonstrate that \textsc{Modularity} is in $\FPT$ when parameterised by the vertex cover number of the input graph.

\begin{thm}\label{thm:vc}
\textsc{Modularity}, parameterised by cardinality of a minimum vertex cover for the input graph~$G$, is in $\FPT$.
\end{thm}

To prove this result, we make use of recent work of Lokshtanov \cite{lokshtanov15} which gives an FPT algorithm for the following problem.

\begin{framed}
\textsc{Integer Quadratic Programming}\\
\textit{Input:} An $n \times n$ integer matrix $Q$, an $m \times n$ integer matrix $A$, and an $m$-dimensional vector~$\mathbf{b}$.\\
\textit{Parameter:} $n + \alpha$, where $\alpha$ is the maximum absolute value of any entry in $A$ or $Q$.\\
\textit{Problem:} Find a vector $\mathbf{x} \in \mathbb{Z}^n$ which minimises $\mathbf{x}^TQ\mathbf{x}$, subject to $A\mathbf{x} \leq \mathbf{b}$.
\end{framed}

Our strategy can be summarised as follows.  We first observe that we may restrict our attention to partitions in which every part intersects the vertex cover.  Moreover, the vertices outside the vertex cover can be classified into at most $2^k$ ``types'' according to their neighbourhood (which by definition must be a subset of the vertex cover).  We then argue that the modularity of a partition depends only on (1) the inherited partition of the vertex cover and (2) the number of (non-vertex-cover) vertices of each type that belong to each of the parts.  Using this characterisation, we can reduce the problem of maximising the modularity to that of solving a collection of instances of \textsc{Integer Quadratic Programming}.

Before embarking on the proof of Theorem \ref{thm:vc}, we introduce some notation.  Suppose that the graph $G = (V,E)$ has $|E| = m$, and that $U = \{u_1,\ldots,u_k\}$ is a vertex cover for~$G$.  Let $\mathcal{P} = \{P_1,\ldots,P_{\ell}\}$ be a partition of $U$, and set $W = V \setminus U$ (so $W$ is an independent set).

We can partition the vertices of $W$ into $2^k$ sets based on their \emph{type}: the type $\tau_U(w) \in \{0,1\}^k$ of a vertex $w \in W$ describes which of the vertices in $U$ are neighbours of $w$.  Formally $\tau_U(w)_j = 1$ if $u_jw\in E(G)$ and $\tau_U(w)_j = 0$ otherwise.  For each $\sigma \in \{0,1\}^k$, we set $S_{\sigma}$ to be the set of all vertices in $W$ with type exactly $\sigma$, that is, $S_\sigma = \{ w\in W : \tau(w)=\sigma\}$.

Now let $\mathcal{A} = \{A_1,\ldots,A_r\}$ be a partition of $V$.  We write $x_{\sigma,i}^{\mathcal{A}}$ for the number of vertices of type $\sigma$ which are assigned to $A_i$, that is, $x_{\sigma,i}^{\mathcal{A}} = |S_{\sigma} \cap A_i|$.  Finally, we introduce 0-1 vectors to encode the sets $P_i \in \mathcal{P}$: for $1 \leq i \leq \ell$, we let $\pi^i \in \{0,1\}^k$ be given by $\pi_j^i = 1$ if $u_j \in P_i$, and $\pi_j^i = 0$ otherwise. An example is given in Figure~\ref{fig:type}.

We now argue that, if the partition $\mathcal{A}$ extends $\mathcal{P}$, we can compute the modularity of $\mathcal{A}$ using only the values $x_{\sigma,i}^{\mathcal{A}}$, together with information about~$\mathcal{P}$.  %This shows that the modularity depends only on the number of vertices of each type assigned to a given partition, and not the assignment of individual vertices.  

\begin{figure}[ht!]
  \begin{center}
    \begin{tikzpicture}[scale=0.8]
      
      %__________________________
         %SECOND PICTURE   
%      \draw [decorate, decoration={brace,amplitude=2pt,mirror},xshift=0pt,yshift=0pt, blue] (-0.4,-0.9) -- (0.4,-0.9) node [black,midway,below,xshift=0cm]  {\textcolor{blue}{\footnotesize $\alpha$}};

      \tikzstyle{vertex}=[circle,fill=black, draw=none, minimum size=4pt,inner sep=2pt]
      \node[vertex,blue] (u1) at (0,4){}; %a1
   \tikzstyle{vertex}=[diamond,fill=black, draw=none, inner sep=1.9pt]      
      \node[vertex,brown] (u2) at (0,0){}; %a4

      \node[left] at (0,4) {$u_1$};
      \node[left] at (-2,2) {};
      \node[left] at (0,0) {$u_2$};
      \node[right] at (2,2) {};

%      \draw (aT)--(aL)--(aB)--(aR)--(aT);
%      \draw (aL)--(aR);
%      \draw (aT) to [out=190,in=170, distance=4cm] (aB);

%      \draw[blue] (aT) to [out=200,in=70] (aL); 
%      \draw[blue] (aR) to [out=200,in=70] (aB); 

      \tikzstyle{vertex}=[circle,fill=black, draw=none, inner sep=2pt]      
      \node[vertex, blue] (vW1) at (4,4.8){}; %type 00
      \node[vertex, blue] (vW2) at (4,4.1){}; %type 00
      \node[vertex, blue] (vW3) at (4,3.4){}; %type 10
      \tikzstyle{vertex}=[rectangle,fill=black, draw=none, inner sep=2.3pt]      
      \node[vertex, red] (vW4) at (4,2.7){}; %type 01
      \node[vertex, red] (vW5) at (4,2.0){}; %type 01  
      \tikzstyle{vertex}=[diamond,fill=black, draw=none, inner sep=1.9pt]      
      \node[vertex, brown] (vW6) at (4,1.3){}; %type 01
      \node[vertex, brown] (vW7) at (4,0.6){}; %type 11
      \node[vertex, brown] (vW8) at (4,-0.1){}; %type 11
      
      \draw (u1) -- (vW3);
      \draw (u1) -- (vW7);
      \draw (u1) -- (vW8);

      \draw (u2) -- (vW4);
      \draw (u2) -- (vW5);
      \draw (u2) -- (vW6);
      \draw (u2) -- (vW7);
      \draw (u2) -- (vW8);

%      \node[vertex] (vLR1p) at (-0.7,1.25){};
%      \draw (vLR1p) -- (vLR1);
      
%      \node[rotate=90,blue] at (3.83,2.4) {\tiny$...$};

     % \begin{scope}[shift={(3.515,2.66)}]
     %       \node[vertex,blue] (vL) at (0,0){};
     %       \node[vertex,blue] (vR) at (0.65,0){};
     %       \draw[blue] (vL)--(vR);
     % \end{scope}

      \draw [decorate, decoration={brace,amplitude=3pt},xshift=-4pt,yshift=0pt] (4.6,4.9) -- (4.6,4.0) node [black,midway,xshift=0.05cm,right]  {\textcolor{black}{\footnotesize $S_{00}$}};

      \draw [decorate, decoration={brace,amplitude=2pt},xshift=-4pt,yshift=0pt, black] (4.6,3.6) -- (4.6,3.2) node [black,midway,xshift=0.05cm,right]  {\textcolor{black}{\footnotesize $S_{10}$}};

      \draw [decorate, decoration={brace,amplitude=3pt},xshift=-4pt,yshift=0pt] (4.6,2.9) -- (4.6,1.1) node [black,midway,xshift=0.05cm,right]  {\textcolor{black}{\footnotesize $S_{01}$}};

      \draw [decorate, decoration={brace,amplitude=3pt},xshift=-4pt,yshift=0pt] (4.6,0.8) -- (4.6,-0.3) node [black,midway,xshift=0.05cm,right]  {\textcolor{black}{\footnotesize $S_{11}$}};

    \end{tikzpicture}%\vspace{-5mm}
  \end{center}
\caption{An example of a graph with vertex cover $U=\{u_1, u_2\}$ and four sets of distinct types indicated for the vertices $W=V\backslash U$. For the vertex partition $\mathcal{A}=\{A_{\color{blue} 1},A_{\color{red} 2},A_{\color{brown} 3}\}$ indicated with  circles ({\color{blue} $\bullet$}), squares ({\tiny \color{red} $\blacksquare$}) and diamonds ({\small \color{brown} $\blacklozenge$}) respectively the only non-zero values of $x^\mathcal{A}_{\sigma, i}$ are: $x^\mathcal{A}_{00,1}=2$, $x^\mathcal{A}_{10,1}=1$, $x^\mathcal{A}_{01,2}=2$, $x^\mathcal{A}_{01,3}=1$ and $x^\mathcal{A}_{11,3}=2$. Note also that $\AA$ extends the partition $\mathcal{P}=\{\{u_1\}, \{u_2\}\}$ of $U$ but not the partition $\mathcal{P}'=\{\{u_1, u_2\}\}$ of $U$. }
\label{fig:type}
\end{figure}

\begin{lma}\label{lma:type-mod}
Let $U = \{u_1,\ldots,u_k\}$ be a vertex cover for $G = (V,E)$, where $|E|=m$, and let $\mathcal{P}$ be a partition of $U$.  If $\mathcal{A}$ is any partition of $V$ which extends $\mathcal{P}$ and has the property that every $A \in \mathcal{A}$ has non-empty intersection with $U$, then
$$q^E_\cA(G)=\frac{1}{m}\sum_{i=1}^{\ell} e(P_i) + \frac{1}{m}\sum_{(\sigma,i)} x_{\sigma,i}^{\mathcal{A}}(  \sigma \cdot \pi^i),$$
and
\begin{align*}
4m^2 q^D_{\AA} = 4\sum_i e(P_i)^2 & + 4\sum_{ (\sigma,i) } x_{\sigma,i}^{\mathcal{A}} e(P_i)  (\sigma\cdot (\mathbf{1}+\pi^i) )\\
 & \qquad + \sum_{(\sigma,i) (\sigma',j)} x_{\sigma,i}^{\mathcal{A}}x_{\sigma',j}^{\mathcal{A}} (\sigma\cdot (\mathbf{1}+\pi^i))(  \sigma' \cdot (\mathbf{1}+\pi^j)).
\end{align*}
\end{lma}
\begin{proof}
Suppose that $\mathcal{P} = \{P_1,\ldots,P_{\ell}\}$ and $\mathcal{A} = \{A_1,\ldots,A_{\ell}\}$, where $P_i \subseteq A_i$ for each $i$; we set $B_i = A_i \cap W$ for each $1 \leq i \leq \ell$.  For any vertex $w \cap B_i$, we have that $e(w,P_i)$ is given by the dot product $\tau(w)\cdot \pi^i$; thus the number of edges between $P_i$ and $B_i$ for each $i$ is given by
\begin{equation}\label{eq.edgebetweeni}
e(P_i,B_i) = \sum_{\sigma \in \{0,1\}^k} x_{\sigma,i}^{\mathcal{A}}(\sigma \cdot \pi^i).
\end{equation}
Since there are no edges inside any set $B_i$, it follows that
\begin{equation*}
e(A_i) = e(P_i) + \sum_{\sigma \in \{0,1\}^k} x_{\sigma,i}^{\mathcal{A}}(\sigma \cdot \pi^i),
\end{equation*}
and hence we can write the edge contribution of $\mathcal{A}$ as
\begin{equation}\label{eq.edgecontVC}
q^E_\cA(G)=\frac{1}{m}\sum_{i=1}^{\ell} e(P_i) + \frac{1}{m}\sum_{(\sigma,i)} x_{\sigma,i}^{\mathcal{A}}(  \sigma \cdot \pi^i).
\end{equation}

Similarly for the degree tax, observe that a vertex $w\in W$ of type $\tau(w)$ has degree $\tau(w)\cdot \mathbf{1}\leq k$, and hence $\vol(B_i)=\sum_{\sigma} x_{\sigma,i}(\sigma\cdot \mathbf{1})$. Notice that $\vol(P_i)=2e(P_i)+e(P_i,B_i)$ and we already have an expression for $e(P_i,B_i)$ in terms of the $x_{\sigma, i}^{\mathcal{A}}$ in~\eqref{eq.edgebetweeni}. Hence, as $\vol(P_i \cup B_i) = \vol(P_i) + \vol(B_i)$, we have

$$4m^2 q^D_{\cA} = \sum_i \vol(P_i \cup B_i)^2 = \sum_i \Big( 2e(P_i)+  \sum_{\sigma} x_{\sigma, i}^{\mathcal{A}} (( \sigma \cdot \pi^i) + (\sigma\cdot \mathbf{1}) ) \Big)^2 $$

and thus rearranging,

\begin{align*} 
4m^2 q^D_{\AA} 
&= 4\sum_i e(P_i)^2 + 4\sum_i e(P_i) \sum_{\sigma} x_{\sigma,i}^{\mathcal{A}}(\sigma\cdot (\mathbf{1}+\pi^i)) \\
& \qquad \qquad + \Big( \sum_{\sigma} x_{\sigma,i}^{\mathcal{A}}(\sigma\cdot (\mathbf{1}+\pi^i) )\Big)^2 \\
&= 4\sum_i e(P_i)^2 + 4\sum_{ (\sigma,i) } x_{\sigma,i}^{\mathcal{A}} e(P_i)  (\sigma\cdot (\mathbf{1}+\pi^i) )\\
& \qquad \qquad + \sum_{(\sigma,i) (\sigma',j)} x_{\sigma,i}^{\mathcal{A}}x_{\sigma',j}^{\mathcal{A}} (\sigma\cdot (\mathbf{1}+\pi^i))(  \sigma' \cdot (\mathbf{1}+\pi^j)), 
\end{align*}
as required. \qed
\end{proof}

We are now ready to prove the main result of this section.

\begin{proof}[Proof of Theorem \ref{thm:vc}]
We will assume that the input to our instance of \textsc{Modularity} is a graph $G=(V,E)$, where $|E| = m$.  We may assume without loss of generality that we are also given as input a vertex cover $U = \{u_1,\ldots,u_k\}$ for~$G$ (as if not we can easily compute one in the allowed time).  We may further assume that~$G$ does not contain any isolated vertices, as we can delete any such vertices (in polynomial time) without changing the value of the maximum modularity (by Fact \ref{fact:isolated}).

Note that the total number of possible partitions of $U$ into non-empty parts is equal to the $k^{th}$ \emph{Bell number}, $B_k$ (and hence is certainly less than $k^k$).  It therefore suffices to describe an fpt-algorithm which determines, given some partition $\mathcal{P}$ of $U$, $$q^{\mathcal{P}}(G) = \max\{q_{\mathcal{A}}(G): \mathcal{A}[U] = \mathcal{P}\}.$$
The maximum modularity of~$G$ can then be calculated by taking $$\max\{q^{\mathcal{P}}(G): \mathcal{P} \text{ is a partition of } U\}.$$

From now on, we consider a fixed partition $\mathcal{P}=\{P_1,\ldots,P_{\ell}\}$ of $U$, and describe how to compute $q^{\mathcal{P}}(G)$.
It follows from {\Fc Facts \ref{fact:connected} and \ref{fact:singlevertex}}, together with the fact that $W$ is an independent set that, if $\mathcal{A} = \{A_1,\ldots,A_j\}$ is a partition of $V$ which achieves the maximum modularity, then every part $A_i$ has non-empty intersection with $U$.  We will call a partition with this properties a $U$-partition of~$G$.  It then suffices to maximise the modularity over all $U$-partitions in order to determine the value of $q^{\mathcal{P}}(G)$.

Now, by Lemma \ref{lma:type-mod}, we know that we can express the modularity of a $U$-partition $\mathcal{A}$ as
\begin{align}
q_{\mathcal{A}}(G) = \frac{1}{m}\sum_{i=1}^{\ell} e(P_i) & + \frac{1}{m}\sum_{(\sigma,i)} x_{\sigma,i}^{\mathcal{A}}(  \sigma \cdot \pi^i) - \frac{1}{m^2}\sum_i e(P_i)^2  \nonumber \\
& - \frac{1}{m^2}\sum_{ (\sigma,i) } x_{\sigma,i}^{\mathcal{A}} e(P_i)  (\sigma\cdot (\mathbf{1}+\pi^i) ) \nonumber \\ 
& - \frac{1}{4m^2}\sum_{(\sigma,i) (\sigma',j)} x_{\sigma,i}^{\mathcal{A}}x_{\sigma',j}^{\mathcal{A}} (\sigma\cdot (\mathbf{1}+\pi^i))(  \sigma' \cdot (\mathbf{1}+\pi^j)). \label{eqn:to-maximise-1}
\end{align}
As we have fixed the partition $\mathcal{P}$, all values $e(P_i)$ can be regarded as fixed constants.  In order to determine the maximum modularity we can obtain with a $U$-partition, we therefore need to find the values of $x_{\sigma,i}^{\mathcal{A}}$ which maximise this expression.

We can rewrite \eqref{eqn:to-maximise-1} as the sum of a constant term, two linear functions $\theta$ and $\phi$ of the $x_{\sigma,i}^{\mathcal{A}}$ and a quadratic function $\psi$ of the $x_{\sigma,i}^{\mathcal{A}}$ (up to scaling by constants):
\begin{align*}
q_{\mathcal{A}}(G) = & \underbrace{\frac{1}{m}\sum_{i=1}^{\ell} e(P_i) - \frac{1}{m^2}\sum_i e(P_i)^2}_{\text{constant}} \\
	& + \frac{1}{m} \underbrace{ \sum_{(\sigma,i)} x_{\sigma,i}^{\mathcal{A}} (\sigma \cdot \pi^i)}_{\theta(\mathcal{A})} - \frac{1}{m^2}\underbrace{\sum_{(\sigma,i)} x_{\sigma,i}^{\mathcal{A}} \, e(P_i)  (\sigma\cdot (\mathbf{1}+\pi^i) )}_{\phi(\mathcal{A})} \\
	& - \frac{1}{4m^2} \underbrace{\sum_{(i,\sigma)(j,\sigma')} x_{\sigma,i}^{\mathcal{A}}x_{\sigma',j}^{\mathcal{A}} (\sigma\cdot (\mathbf{1}+\pi^i))(  \sigma' \cdot (\mathbf{1}+\pi^j)) }_{\psi(\mathcal{A})}.
\end{align*}
To find the maximum value of $q_{\mathcal{A}}(G)$ over all $U$-partitions it therefore suffices to determine, for all possible values of $\theta(\mathcal{A})$ and $\phi(\mathcal{A})$, the minimum possible value of $\psi(\mathcal{A})$.  Before describing how to do this, we observe that the number of combinations of possible values for $\theta(\mathcal{A})$ and $\phi(\mathcal{A})$ and is not too large.  Note that $0 \leq \sum_{\sigma,i} x_{\sigma,i}^{\mathcal{A}} (\sigma \cdot \pi^i) < nk$, and $0 \leq \sum_{\sigma,i} x_{\sigma,i}^{\mathcal{A}} e(P_i)(\sigma \cdot (\mathbf{1} + \pi^i)) < n \binom{k}{2} 2k < nk^3$, so the number of possible pairs $\left(\theta(\mathcal{A}),\phi(\mathcal{A})\right)$ is at most $n^2k^4$.  Thus, if we know the minimum possible value of $\psi(\mathcal{A})$ corresponding to each possible pair $\left(\theta(\mathcal{A}),\phi(\mathcal{A})\right)$, we can compute the maximum modularity achieved by any $U$-partition $\mathcal{A}$ such that $\left(\theta(\mathcal{A}),\phi(\mathcal{A})\right) = (y,z)$, and maximising over the polynomial number of possible pairs $(y,z)$ will give $q^{\mathcal{P}}(G)$.

Now, given a possible pair of values $(y,z)$ for $\left(\theta(\mathcal{A}),\phi(\mathcal{A})\right)$, we describe how to compute 
$$\min\{\psi(\mathcal{A}): \mathcal{A} \text{ is a $U$-partition with } \theta(\mathcal{A}) = y \text{ and } \phi(\mathcal{A}) = z\}.$$  
Our strategy is to express this minimisation problem as an instance of \textsc{Integer Quadratic Programming} and then apply the FPT algorithm of \cite{lokshtanov15}.

In this instance, we have $n = \ell 2^k \leq k2^k$, and our vector of variables $\mathbf{x} = (x_1,\ldots,x_n)^T$ is given by
$$x_i = x_{\left(\sigma_{i \mod 2^k}\right),\left\lceil i/ 2^k \right\rceil}^{\mathcal{A}},$$
where $\sigma_1,\ldots,\sigma_{2^k}$ is a fixed enumeration of all vectors in $\{0,1\}^k$.  The matrix $Q$ expresses the value of $\psi(\mathcal{A})$ in terms of $\mathbf{x}$:  if we set $Q = \{q_{i,j}\}$ where 
$$q_{i,j} = \left( \sigma_{\left(i \mod 2^k\right)} \cdot \left(\mathbf{1} + \pi^{\left\lceil i/ 2^k \right\rceil}\right)\right) \left(\sigma_{\left(j \mod 2^k \right)} \cdot \left(\mathbf{1} + \pi^{\left\lceil j / 2^k \right\rceil}\right) \right),$$
then it is easy to see that $\psi(\mathcal{A}) = \mathbf{x}^T Q \mathbf{x}$.  Note also that the maximum absolute value of any entry in $Q$ is at most $4k^2$.

We now use the linear constraints to express the conditions that
\begin{enumerate}
\item $\theta(\mathcal{A}) = y$, 
\item $\phi(\mathcal{A}) = z$, and
\item the values $x_{i,\sigma}$ correspond to a valid $U$-partition $\mathcal{A}$.
\end{enumerate}
The first of these conditions can be expressed as a single linear constraint:
$$\sum_{(\sigma,i)} x_{\sigma,i}^{\mathcal{A}} (\sigma \cdot \pi^i) = y,$$
or equivalently $\mathbf{a}_1 \mathbf{x} = y$ where $\mathbf{a}_1$ is the $1 \times n$ row vector with $i^{th}$ entry equal to 
$$\sigma_{\left( i \mod 2^k \right)} \cdot \pi^{\left \lceil i/2^k \right \rceil}.$$
We can similarly express the second condition as a single linear constraint:
$$\sum_{(\sigma,i)} x_{\sigma,i}^{\mathcal{A}} \, e(P_i)  (\sigma\cdot (\mathbf{1}+\pi^i) ) = z,$$
or equivalently $\mathbf{a}_2 \mathbf{x} = z$, where $\mathbf{a}_2$ is the $1 \times n$ row vector with $i^{th}$ entry equal to 
$$e\left(P_{\left\lceil i/2^k \right\rceil}\right) \left(\sigma_{\left( i \mod 2^k \right)} \cdot \left( \mathbf{1} + \pi^{\left \lceil i / 2^k \right \rceil} \right) \right).$$
Note that every entry in the vectors $\mathbf{a}_1$ and $\mathbf{a}_2$ has absolute value no more than $2k^3$.  For the third condition, note that the values $x_{i,\sigma}$ correspond to a valid $U$-partition if and only if every $x_{i,\sigma}$ is non-negative, and for each $\sigma$ we have $\sum_{i = 1}^{\ell} x_{i,\sigma}^{\mathcal{A}} = |S_{\sigma}|$.

We can therefore express all three conditions in the form $A\mathbf{x} = \mathbf{b}$, where $A$ is a $\left(4 + (\ell+1)2^k\right) \times n$ and $\mathbf{b}$ is a $\left(4 + (\ell+1)2^k\right)$-dimensional vector (notice that we use two inequalities to express each of the linear equality constraints).

Altogether, this means that the solution to this \textsc{Integer Quadratic Programming} instance will determine the values of $x_{i,\sigma}^{\mathcal{A}}$ which minimize (out of all values corresponding to some $U$-partition $\mathcal{A}$) the value of $\psi(\mathcal{A})$, subject to the additional requirement that $\theta(\mathcal{A}) = y$ and $\phi(\mathcal{A}) = z$.  Note that the number of variables $n$ is at most $k2^k$ and the largest absolute value of any entry in $A$ or $Q$ is at most $2k^3$, so the parameter in the instance of \textsc{Integer Quadratic Programming} is bounded by a function of $k$.  This completes the proof. \qed
\end{proof}

We note the algorithm described can easily be modified to output an optimal partition.

\subsection{Parameterisation by treewidth}
\label{sec:tw}

In this section we demonstrate that \textsc{Modularity}, when parameterised by the treewidth of the input graph~$G$, belongs to $\XP$ and so is solvable in polynomial time on graph classes whose treewidth is bounded by some fixed constant.  We further show that for any fixed $\eps>0$ there is an FPT-algorithm, parameterised by treewidth, which computes a factor $(1-\eps)$-approximation; i.e.\ returning a value between $(1-\eps)\q$ and $\q$ where $\q$ is the maximum modularity of the graph.

\begin{thm}\label{thm:tw-xp}
\textsc{Modularity} parameterised by the treewidth of the input graph~$G$ is in~$\XP$.
\end{thm}
\begin{proof}
As the proof makes use of standard dynamic programming techniques on tree decompositions, we only give an outline proof here.  Suppose that~$G$ has $n$ vertices and $m$ edges, and has treewidth $k$.  We will assume that we are given a nice tree decomposition $(T,\mathcal{D})$ (where $T$ is a tree and $\mathcal{D} = \{\mathcal{D}(t): t \in V(T)\}$) of~$G$, of width $k$, as part of the input (if not we can compute one in FPT time).

The proof relies heavily on Fact \ref{fact:connected}.  This means we can compute the optimum modularity without considering partitions that induce disconnected subgraphs; hence, for any node $t \in V(T)$, we need only consider partitions $\mathcal{A}$ with the property that, if $A \in \mathcal{A}$ does not intersect $\mathcal{D}(t)$, then all vertices in $A$ only appear in bags indexed by nodes in precisely one connected component of $T \setminus t$.

We compute the modularity by working upwards from the leaves in the standard way.  As we do this, we need to keep track of relevant statistics for the parts that intersect the current bag (\emph{liquid} parts) and also the total contribution to the modularity from the parts (\emph{frozen} parts) which contain only vertices from bags indexed by descendants of the current node (and so by the reasoning above cannot accept more vertices from elsewhere in the graph).

For any node $t \in V(T)$, a valid \emph{state} of $t$ consists of the following: 
\begin{enumerate}
\item a partition $\mathcal{P}$ of $\mathcal{D}(t)$;
\item a function $\alpha : \mathcal{P} \rightarrow [m]$ such that $\alpha(P_i) \geq e(P_i)$ for each $P_i \in \mathcal{P}$;
\item a function $\beta : \mathcal{P} \rightarrow [2m]$ such that $\beta(P_i) \geq \vol(P_i)$ for each $P_i \in \mathcal{P}$.
\end{enumerate}
Here $\mathcal{P}$ records the restriction of a partition to $\mathcal{D}(t)$, $\alpha$ keeps track of the number of edges captured so far in each of the liquid parts, and $\beta$ keeps track of the volume so far of each of the liquid parts.  Notice that the total number of possible states for any node $t$ is at most $(k+1)^{(k+1)} \cdot m^{(k+1)} \cdot (2m)^{(k+1)} = m^{\mathcal{O}(k)}$.

For each possible state of a node $t$, we need to keep track of the maximum contribution to modularity from frozen parts we can achieve consistent with the liquid parts having the specified state: this is done with a function $\sigma_t$, the \emph{signature} of $t$.  Given any state $(\mathcal{P},\alpha,\beta)$ of $t$, we first define a \emph{$(t,\mathcal{P},\alpha,\beta)$-partition} to be any partition $\mathcal{A}$ of $V_t$ such that:
\begin{enumerate}
\item $\mathcal{P} = \mathcal{A}[\mathcal{D}(t)]$;
\item for all $A \in \mathcal{A}$ with $A \cap \mathcal{D}(t) \neq \emptyset$:
\begin{itemize}
\item $\alpha\left(A \cap \mathcal{D}(t)\right) = e(A)$, and
\item $\beta\left(A \cap \mathcal{D}(t) \right) = \vol(A)$.
\end{itemize}
\end{enumerate}
We then set
\begin{align*}
\sigma_t(\mathcal{P},\alpha,\beta) = \max \bigg\{ \frac{1}{m} \sum_{B \in \mathcal{B}} e(B) - \frac{1}{m^2} \sum_{B \in \mathcal{B}} & \vol(B)^2 : \\
& \mathcal{A} \text{ is a $(t, \mathcal{P}, \alpha, \beta)$-partition and } \\
& \mathcal{B} = \{A \in \mathcal{A}: A \cap \mathcal{D}(t) = \emptyset\} \bigg\}.
\end{align*}
Throughout the proof we adopt the convention that the maximum value of an empty set is $- \infty$.

It is clear that, with knowledge of $\sigma_r$ for the root $r$ of the tree decomposition, we can easily determine the maximum modularity of~$G$.  It therefore remains to outline how we compute $\sigma_t$ for the four types of node in the nice tree decomposition, using only information about the values of $\sigma_{t'}$ where $t'$ is a child of $t$.  We begin by observing that if $t$ is a leaf node then we can exhaustively consider all possibilities in time depending only on $k$.

Now suppose $t$ is an introduce node with child $t'$, where $\mathcal{D}(t) = \mathcal{D}(t') \cup \{v\}$.  Given any state $(\mathcal{P},\alpha,\beta)$ of $t$, we say that a state $(\mathcal{P}',\alpha',\beta')$ of $t'$ is introduce-compatible with $(\mathcal{P},\alpha,\beta)$ if:
\begin{itemize}
\item $\mathcal{P}' = \mathcal{P} \setminus \{v\}$;
\item for every $P \in \mathcal{P}$, if $v \notin P$ then $\alpha'(P) = \alpha(P)$, and if $v \in P$ (but $P \setminus \{v\} \neq \emptyset$) then $\alpha'(P) = \alpha(P) - |\{u \in P: uv \in E(G)\}|$;
\item for every $P \in \mathcal{P}$, if $v \notin P$ then $\beta'(P) = \beta(P)$, and if $v \in P$ (but $P \setminus \{v\} \neq \emptyset$) then $\beta'(P) = \beta(P) - d(v)$.
\end{itemize}
It then follows that $\sigma_t(\mathcal{P},\alpha,\beta)$ is equal to
$$\max \{\sigma_{t'}(\mathcal{P}',\alpha',\beta'): (\mathcal{P}',\alpha',\beta') \text{ is introduce-compatible with } (\mathcal{P},\alpha,\beta)\}.$$

Next, suppose that $t$ is a forget node with child $t'$, where $\mathcal{D}(t) = \mathcal{D}(t') \setminus \{v\}$.  Given any state $(\mathcal{P},\alpha,\beta)$ of $t$, we define two functions $\sigma_t^1$ and $\sigma_t^2$; these functions correspond to the case where one of the parts that is liquid at $t'$ becomes frozen at $t$ (if $v$ was the last vertex in its part), and the case where all parts that are liquid at $t'$ remain liquid at $t$, respectively.  We set 
\begin{align*}
\sigma_t^1 (\mathcal{P}, \alpha, \beta) = \max \bigg\{ \sigma_{t'}(\mathcal{P}',\alpha',\beta')&  + \frac{1}{m} \alpha'\left(\{v\}\right) - \frac{1}{4m^2}\beta'\left(\{v\}\right)^2:  \\
& \mathcal{P}' = \mathcal{P} \cup \{v\} \text{ and, for all } P \in \mathcal{P},\\
&\alpha'(P) = \alpha(P) \text{ and } \beta'(P) = \beta(P) \bigg\},
\end{align*}
and
\begin{align*}
\sigma_t^2 (\mathcal{P},\alpha, \beta) = \max \bigg\{\sigma_{t'}(\mathcal{P}',\alpha',\beta'): \; & \mathcal{P} = \mathcal{P}' \setminus \{v\}, |\mathcal{P}'| = |\mathcal{P}| \text{ and, } \\
& \text{for all } P \in \mathcal{P}',\alpha'(P) = \alpha(P \setminus v) \\
& \text{and } \beta'(P) = \beta(P \setminus v) \bigg\}.
\end{align*}
We then see that 
$$\sigma_t(\mathcal{P},\alpha,\beta) = \max \left\{ \sigma_t^1(\mathcal{P},\alpha,\beta), \sigma_t^2(\mathcal{P},\alpha,\beta)\right\}.$$

Finally, suppose that $t$ is a join node with children $t_1$ and $t_2$, where $\mathcal{D}(t_1) = \mathcal{D}(t_2) = \mathcal{D}(t)$.  In this case we see that 
\begin{align*}
\sigma_t(\mathcal{P},\alpha,\beta) = \max \bigg\{  \sigma_{t_1}(\mathcal{P},\alpha_1,\beta_1) + & \sigma_{t_2}(\mathcal{P},\alpha_2,\beta_2): \text{ for all } P \in \mathcal{P}, \\
& \alpha(P) = \alpha_1(P) + \alpha_2(P) - e(P) \text{ and } \\
& \beta(P) = \beta_1(P) + \beta_2(P) - \vol(P) \bigg\}.
\end{align*} \qed
\end{proof}

To obtain our FPT approximation result, we use a very similar approach; the key is to restrict our attention to partitions with only a constant number of parts.  For any constant $c \in \mathbb{N}$, we write $q_{\leq c}(G)$ for the maximum modularity for~$G$ achievable with a partition into at most $c$ parts, that is
$$q_{\leq c}(G) = \max_{|\mathcal{A}| \leq c} q_{\mathcal{A}}(G).$$
We refer to the problem of deciding whether $q_{\leq c}(G) \geq q$ for a given input graph~$G$ and constant $q \in [0,1]$ as $c$-\textsc{Modularity}.  We now argue that $c$-\textsc{Modularity} is in $\FPT$ parameterised by the treewidth of the input graph.  The crucial difference from our XP algorithm above is the fact that, when we fix the number of parts in the partition, we can no longer assume that every part is connected.  However, if the maximum number of parts $c$ is a constant, we can keep track of the necessary statistics for every possible part, not just those that intersect the bag under consideration.

\begin{lma}\label{lma:const-parts}
$c$-\textsc{Modularity} is in $\FPT$ when parameterised by the treewidth of the input graph.
\end{lma}
\begin{proof}
The strategy is broadly the same as that used in the proof of Theorem \ref{thm:tw-xp}, however when the number of parts is fixed we can no longer assume that every part in the optimal partition is connected.  Thus, instead of recording statistics relating to each part that intersects the bag currently under consideration, we keep track of the same statistics for each of the $c$ (possibly empty) parts allowed in the partition.  Formally, for any node $t \in V(T)$, a valid state of $t$ consists of:
\begin{enumerate}
\item a function $\pi: \mathcal{D}(t) \rightarrow [c]$;
\item a function $\alpha \colon [c] \rightarrow [m]$ such that $\alpha(i) \geq e(\pi^{-1}(i))$ for all $i \in [c]$;
\item a function $\beta \colon [c] \rightarrow [2m]$ such that $\beta(i) \geq \vol(\pi^{-1}(i))$ for all $i \in [c]$.
\end{enumerate}
Here $\pi$ records the mapping of vertices of $\mathcal{D}(t)$ to the $c$ possible parts, $\alpha$ keeps track of the number of edges captured so far in each of the $c$ parts, and $\beta$ the volume so far of each part.  Notice that the number of possible states for any node $t$ is at most $c^{k+1} \cdot m^{c} \cdot (2m)^c = c^{k+1} m^{\mathcal{O}(c)}$.

Given any state $(\pi,\alpha,\beta)$ of $t$, we define a \emph{$(t,\pi,\alpha,\beta)$-partition} to be any partition $\mathcal{A} = \{A_1,\ldots,A_c\}$ of $V_t$ such that:
\begin{enumerate}
\item $v \in A_{\pi(v)}$ for each $v \in \mathcal{D}(t)$;
\item for each $i \in [c]$:
\begin{itemize}
\item $\alpha(i) = e(A_i)$, and
\item $\beta(i) = \vol(A_i)$.
\end{itemize}
\end{enumerate}
We then set 
\begin{equation*}
\theta_t(\pi,\alpha,\beta) = \begin{cases}
								1 	&\text{if there exists a $(t,\pi,\alpha,\beta)$-partition of $V_t$,} \\
								0   &\text{otherwise.}
							 \end{cases}
\end{equation*}
It is clear that, if $r$ is the root of the tree decomposition,
$$q_{\leq c}(G) = \max_{\theta_r(\pi,\alpha,\beta) = 1} \left\lbrace \frac{1}{m} \sum_{i = 1}^c \alpha(i) + \frac{1}{m^2} \sum_{i = 1}^c \beta(i)^2 \right\rbrace.$$
Thus it suffices to compute all values of $\theta_r$.  Note that if $t$ is a leaf node we can consider all possibilities in time depending only on $k$ and $c$; we now outline how to compute the values of $\theta_t$ for a node $t$, given the values for its children.

Suppose first that $t$ is an introduce node with child $t'$, where $\mathcal{D}(t) = \mathcal{D}(t') \cup \{v\}$.  Given any state $(\pi,\alpha,\beta)$ of $t$, we say that a state $(\pi',\alpha',\beta')$ of $t'$ is introduce-compatible with $(\pi,\alpha,\beta)$ if:
\begin{itemize}
\item $\pi' = \pi|_{\mathcal{D}(t')}$;
\item for every $i \in [c]$, if $\pi(v) \neq c$ then $\alpha'(i) = \alpha(i)$, and if $\pi(v) = i$ then $\alpha'(i) = \alpha(i) - |\{u \in \pi^{-1}(i): uv \in E(G)\}|$;
\item for every $i \in [c]$, if $\pi(v) \neq i$ then $\beta'(i) = \beta(i)$, and if $\pi(v) = P$ then $\beta'(i) = \beta(i) - d(v)$.
\end{itemize}
It then follows that $\theta_t(\mathcal{P},\alpha,\beta)$ is equal to
$$\max \{\theta_{t'}(\pi',\alpha',\beta'): (\pi',\alpha',\beta') \text{ is introduce-compatible with } (\pi,\alpha,\beta)\}.$$

Next, suppose that $t$ is a forget node with child $t'$, where $\mathcal{D}(t) = \mathcal{D}(t') \setminus \{v\}$.  In this case we have
$$\theta_t (\pi,\alpha,\beta) = \max \{\theta_{t'}(\pi',\alpha',\beta'): \; \pi = \pi'|_{\mathcal{D}(t)}, \alpha' = \alpha \text{ and } \beta' = \beta \}.$$

Finally, suppose that $t$ is a join node with children $t_1$ and $t_2$, where $\mathcal{D}(t_1) = \mathcal{D}(t_2) = \mathcal{D}(t)$.  In this case we see that 
\begin{align*}
\theta_t(\pi,\alpha,\beta) = \max \bigg\{ & \theta_{t_1}(\pi,\alpha_1,\beta_1) \cdot \theta_{t_2}(\pi,\alpha_2,\beta_2): \\
& \forall i \in [c], \alpha(i) = \alpha_1(i) + \alpha_2(i) - e(\pi^{-1}(i)) \\
& \text{ and } \beta(i) = \beta_1(i) + \beta_2(i) - \vol(\pi^{-1}(i)) \bigg\}.
\end{align*} \qed
\end{proof}

Recall (Fact \ref{fact:approxbycALT}) that $\q(G) \geq q_{\leq c}(G) > \q(G)\big(1-\frac{1}{c}\big)$; thus, for any constant $\epsilon > 0$, we obtain a factor $(1-\eps)$-approximation by solving $\lceil \frac{1}{\epsilon} \rceil$-\textsc{Modularity}.  This immediately gives the following result.

\begin{cor}
Given any constant $\epsilon > 0$, there is an FPT-algorithm, parameterised by the treewidth of the input graph~$G$, that returns a partition $\AA$ with $q_{\AA}(G)>(1-\eps)\q(G)$.
\end{cor}

We conclude this section by noting that sparse graphs, in particular graphs~$G$ with low tree width, $\tw(G)$, and maximum degree, $\triangle(G)$, can have high maximum modularity. In particular Theorem~1.11 of~\cite{treelike} shows $\q(G)\geq 1-2((\tw(G)+1)\triangle(G)/|E(G)|)^{1/2}$.

\subsection{Parameterisation by max leaf number}

In this section we demonstrate that \textsc{Modularity} can be solved in time linear in the number of connected subgraphs of the input graph~$G$; as a consequence of this result, we deduce that the problem belongs to $\XP$ when parameterised by the max leaf number of~$G$.

\begin{thm}\label{thm:conn-subgs}
Let~$G$ be a graph on $n$ vertices with $m$ edges and at most $h$ connected subgraphs.  Then \textsc{Modularity} can be solved in time $\mathcal{O}(h^2n)$.
\end{thm}
\begin{proof}
We will assume without loss of generality (by Fact \ref{fact:isolated}) that~$G$ contains no isolated vertices.  For any induced subgraph $H$ of~$G$, and partition $\mathcal{A}_H$ of $V(H)$, we write
$$q_{\mathcal{A}_H}(H,G) = \frac{1}{m} \sum_{A \in \mathcal{A}_H} e(A) - \frac{1}{4m^2} \sum_{A \in \mathcal{A}_H} \vol(A)^2,$$
where $\vol(A)$ denotes the volume of $A$ in $G$.  We then set 
$$q^*(H,G) = \max_{\mathcal{A}_H} q_{\mathcal{A}_H}(H,G),$$ 
where the maximum is taken over all partitions $\mathcal{A}_H$ of $V(H)$.  Thus, $q^*(H,G)$ can be seen as the maximum possible contribution of parts contained in $H$ to the modularity of~$G$, if we only consider partitions of $V(G)$ such that every part is either completely contained in $V(H)$ or does not intersect $V(H)$.

Let $H$ be a connected subgraph of~$G$.  Then, for any partition $\mathcal{A}_H$ of $V(H)$ with $|\mathcal{A}_H| > 1$, such that each part induces a connected subgraph, it is clear that there exists a partition $(X,Y)$ of $V(H)$ into two nonempty sets such that $H[X]$ and $H[Y]$ are both connected, and every element of $\mathcal{A}_H$ is completely contained in either $X$ or $Y$.  Conversely, if $(X,Y)$ is a partition with this property it is immediate that partitions of $X$ and $Y$ can be combined to give a partition of $V(H)$.  For any connected graph $H$, we write $\mathcal{P}(H)$ for the set of all partitions $(X,Y)$ of $V(H)$ into two non-empty sets such that $G[X]$ and $G[Y]$ are both connected.  Since we need only consider partitions in which every part induces a connected subgraph (by Fact \ref{fact:connected}), it follows that
\begin{align}
\!\!\!\!q^*(H,G) = \max \bigg\lbrace\! \Big(\frac{1}{m}e(H) - &\frac{1}{m^2} \vol(H)^2 \Big), \nonumber \\
&\!\!\max_{(X,Y)\in \mathcal{P}(H)}\!\! \Big\lbrace q^*(G[X],G) + q^*(G[Y],G)\Big\rbrace\! \bigg\rbrace,
\label{eqn:split-conn}
\end{align}
again adopting the convention that the maximum, taken over an empty set, is equal to $- \infty$.

By assumption,~$G$ has only $h$ connected induced subgraphs.  We note that, with suitable data structures, we can compute a list of all such subgraphs in time $\mathcal{O}(nh)$.  To enumerate all connected induced subgraphs containing the vertex $v$, we can explore a search tree as follows: we associate the pair $(\{v\},V(G) \setminus \{v\})$ with the root and, on reaching a node associated with the pair $(U,W)$, we select an arbitrary vertex $x \in W$ such that $N(x) \cap U \neq \emptyset$ (if such a vertex exists), and create two child nodes associated with $(U \cup \{x\}, W \setminus \{x\})$ and $(U, W \setminus \{x\})$ respectively.  When this process terminates, the vertex-set of every connected induced subgraph appears as the first element of the tuple for exactly one leaf node.  Repeating the process for each vertex in the graph (after deleting those starting vertices already considered) will produce a list of all connected induced subgraphs.  

From now on we will assume that we have computed a list $H_1,\ldots,H_h$ of all connected induced subgraphs of $G$; without loss of generality we may further assume that these subgraphs are listed in non-decreasing order of their number of vertices.  In particular, this means that there is no connected induced subgraph that is strictly contained in $H_1$, so $\mathcal{P}(H_1) = \emptyset$ and $q^*(H_1,G) = \frac{1}{m}e(H_1) - \frac{1}{m^2} \vol(H_1)^2$.  We can reformulate \eqref{eqn:split-conn} as follows:
\begin{align*}
q^*(H_j,G) = \max \Bigg\lbrace\! \Big(\frac{1}{m}&e(H_j) - \frac{1}{m^2} \vol(H_j)^2 \Big),\\
&\max_{\substack{i < j \\ V(H_i) \subset V(H_j) \\ H_j\setminus V(H_i) \text{ connected}}} \!\!\!\!\!\!\!\!\!\!\!\!\!\!\!\!\Big\lbrace q^*(H_i,G) + q^*(H_j \setminus V(H_i) ,G)\Big\rbrace\! \Bigg\rbrace.
\end{align*}

Note that, if $H_j \setminus V(H_i)$ is connected, then $H_j \setminus V(H_i)$ is $H_{\ell}$ for some $\ell < j$.  Thus, if we know the values $q^*(H_1,G),\ldots,q^*(H_{j-1},G)$, we can compute $q^*(H_j,G)$ in time $\mathcal{O}(j|H_j|)$.  It follows that, by considering the connected subgraphs $H_1,\ldots,H_h$ in order, we can compute $q^*(H)$ for every connected induced subgraph in time $\mathcal{O}(h^2n)$.

Now suppose that~$G$ has connected components $C_1,\ldots,C_{\ell}$, where $V(C_i) = V_i$ for each $i$.  By Fact \ref{fact:connected} (see also Lemma~1.6.2 of \cite{thesis}), we can restrict our attention to partitions $\mathcal{A}$ of $V(G)$ such that every part is completely contained in some $V_i$,
$$q^*(G) = \sum_{i = 1}^{\ell} q^*(C_i,G).$$
Since each connected component $C_i$ is a connected induced subgraph of~$G$, it occurs in the list $H_1,\ldots,H_h$ of connected induced subgraphs.  Thus, once we have computed $q^*(H,G)$ for each connected induced subgraph $H$, we can immediately determine $q^*(G)$ by summing the appropriate values.  Hence the overall time required to compute $q^*(G)$ is $\mathcal{O}(h^2n)$. \qed
\end{proof}

It is known that, if the max leaf number of~$G$ is $c$, then~$G$ is a subdivision of some graph $H$ on at most $4c$ vertices \cite{estivill05}; a graph on $n$ vertices that is a subdivision of such a graph $H$ has at most $2^{4c}n^{(4c)^2}$ connected subgraphs (once we have decided which branch vertices belong to a subgraph, it remains only to decide where to cut each path from one of the chosen branch vertices to one we have not chosen).  Thus, if the max leaf number of~$G$ is bounded by a constant it follows that~$G$ has at most a polynomial number of connected subgraphs, and the following result is an immediate consequence of Theorem~\ref{thm:conn-subgs}.

\begin{cor}
\textsc{Modularity} is in $\XP$ when parameterised by the max leaf number of the input graph~$G$.
\end{cor}

We conjecture that this result is not optimal, and that \textsc{Modularity} is in fact in $\FPT$ with respect to this parameterisation.

\section{Hardness results}
\label{sec:hardness}

In this section we complement our positive result about the FPT approximability of the problem parameterised by treewidth by demonstrating that computing the exact value of the maximum modularity is hard even in a more restricted setting.

\begin{thm}\label{thm:hard}
\textsc{Modularity}, parameterised simultaneously by the pathwidth and the size of a minimum feedback vertex set for the input graph, is $\W[1]$-hard.
\end{thm}

Our proof of this result relies on the hardness of the following problem.

\begin{framed}
\noindent
\textsc{Equitable Connected Partition (ECP)}\\
\textit{Input:} A graph $G = (V,E)$ and $r \in \mathbb{N}$.\\
\textit{Question:} Is there a partition of $V$ into $r$ classes $V_1,\ldots,V_r$ such that $|V_i| - |V_j| \leq 1$ for all $1 \leq i < j \leq r$, and the induced subgraph $G[V_i]$ is connected for each $i \in 1,\ldots,r$?
\end{framed}

The parameterised complexity of ECP was investigated thoroughly in \cite{enciso09}.  Among other results, the problem is shown to be $\W[1]$-hard even when parameterised simultaneously by $r$, $\pw(G)$ and $\fvs(G)$.  In proving this hardness result, the authors implicitly consider the following variation of ECP.

\begin{framed}
\noindent
\textsc{Anchored Equitable Connected Partition (AECP)}\\
\textit{Input:} A graph $H = (V_H,E_H)$, and a set of distinguished \emph{anchor vertices} $a_1,\ldots,a_r \in V$.\\
\textit{Question:} Is there a partition of $V_H$ into $r$ classes $V_1,\ldots,V_r$ such that $a_i \in V_i$ for all $i$, $\left||V_i| - |V_j|\right| \leq 1$ for all $1 \leq i < j \leq r$, and the induced subgraph $G[V_i]$ is connected for each $i \in 1,\ldots,r$?
\end{framed}
From the proof of \cite[Theorem 1]{enciso09} we can extract the following statement about the hardness of AECP.

\begin{lma}[\cite{enciso09}, implicit in proof of Theorem 1]\label{lma:aecp-hard}
AECP is $\W[1]$-hard, parameterised simultaneously by $\pw(H)$ and $\fvs(H)$, even if the following conditions hold simultaneously:
\begin{enumerate}
\item $H$ is connected;
\item the graph $H'$ obtained from $H$ by deleting all vertices of degree one is a subdivision of a 3-regular graph $\tilde{H}$;
\item the branch vertices of $H'$ (i.e.\ vertices of $\tilde{H}$) are precisely the anchor vertices $a_1,\ldots,a_r$;
\item $r\geq 4$ is even and divides $|V_H|$;
\item $H \setminus \{a_1,\ldots,a_r\}$ is a disjoint union of isolated vertices and paths with pendant edges.
\end{enumerate}
\end{lma}

In the proof of Theorem \ref{thm:hard}, it is useful to analyse the `per unit modularity {\Fc deficit}' $f_m(B)$ of vertex subsets $B$.  {\Fc For $m\geq 1$ and vertex subset $B$ with $\vol(B)\geq 1$ we define \begin{equation}\label{eq.perunitdeficit}f_m(B)=\frac{\partial(B)}{\vol(B)}+\frac{\vol(B)}{2m}. \end{equation}

Intuitively, minimising the per unit modularity deficit $f_m(B)$ maximises the modularity (see \eqref{eq.exact} for a precise statement). Hence, loosely,}
the following lemma says that if we are restricted to parts $B$ with $\delta(B)=4$ the modularity maximising volume is $\vol(B)=2\sqrt{2m}$. Moreover, while it would usually be better to take parts with $\delta(B)<4$ these parts are actually worse (i.e.\ higher $f_m(B)$ value) if their volumes are too big or too small. The function $f_m(B)$ plays a similar role to the \emph{n-cost} in Proposition~1 of~\cite{treelike}.

\begin{lma}\label{lem.maxf4} Let $m\geq 1$, $\vol(B)\geq 1$ and let $f_m(B)$ be as defined in \eqref{eq.perunitdeficit}.
Then the following properties hold:

\noindent
\emph{0:} if $\partial(B)=0$ and $\vol(B)>4 
\sqrt{2m}$ then $f_m(B)>2\sqrt{2/m}$.\\
\emph{1:} if $\partial(B)=1$ and $\vol(B)>3.7321 
\sqrt{2m}$ or 
$\vol(B)<0.2679
\sqrt{2m}$ then $f_m(B)>2\sqrt{2/m}$.\\
\noindent\emph{2:} if $\partial(B)=2$ and $\vol(B)>3.4143 
\sqrt{2m}$ or 
$\vol(B)<0.5857
\sqrt{2m}$ then $f_m(B)>2\sqrt{2/m}$.\\
\noindent\emph{3:} if $\partial(B)=3$ and $\vol(B)>3\sqrt{2m}$ or $\vol(B)<\sqrt{2m}$ then $f_m(B)>2\sqrt{2/m}$.\\ 
\noindent\emph{4:} if $\partial(B)=4$ and $\vol(B)\geq 2\sqrt{2m}$ then $f_m(B) \geq 2\sqrt{2/m}$ with equality iff $\vol(B)=2\sqrt{2m}$.\\
\noindent\emph{5:} if $\partial(B)\geq 5$ then $f_m(B)>2\sqrt{2/m}$.
\end{lma}
\begin{proof} 

Fix a vertex set $B$ with a constant number, $\ell$, of edges to the rest of the graph (so $\partial(B)=\ell$). {\Fc For $\ell=0$ one can check that directly that if $\vol(B)>4\sqrt{2m}$ then $f_m(B)>2\sqrt{2/m}$ which establishes part 0.} Thus we may assume $\ell\geq 1$. By definition,
$f_m(B)=\ell/\vol(B)+\vol(B)/(2m)$ and so

\begin{equation}\label{eq.forFA5}f_m(B)\geq 2\sqrt{\frac{2}{m}}\;\; \Leftrightarrow \;\; \left( \frac{\vol(B)}{\sqrt{2m}} -2\right)^2 \geq 4-\ell.\end{equation}

Hence for $\ell=4$ we get equality iff $\vol(B)=2\sqrt{2m}$ which immediately implies part 4 of the lemma. Also for $\ell\geq 5$ the RHS of~\eqref{eq.forFA5} is negative which gives part 5 of the lemma. It remains to prove parts 1, 2 and 3 of the lemma.

Now suppose $\ell\in\{1,2,3\}$ then $\sqrt{4-\ell}$ is real and so we may rearrange as the difference of two squares,

$$f_m(B)>2\sqrt{\frac{2}{m}} \;\; \Leftrightarrow \;\; \left( \frac{\vol(B)}{\sqrt{2m}} -2-\sqrt{4-\ell}\right)\left( \frac{\vol(B)}{\sqrt{2m}} -2+\sqrt{4-\ell}\right) >0.$$

Observe $f_m(B)>2\sqrt{2/m}$ if the terms in the product above are either both positive or both negative. Hence $f_m(B)>2\sqrt{2/m}$ if 

$$\vol(B)>2\sqrt{2m}\left(1+\tfrac{1}{2}\sqrt{4-\ell}\right)\;\;\;\;\;\; \mbox{ or } \;\;\;\;\;\; \vol(B)<2\sqrt{2m}\left(1-\tfrac{1}{2}\sqrt{4-\ell}\right).$$

Therefore for $\ell=1$ we get $f_m(B)>2\sqrt{2m}$ if $\vol(B)>\sqrt{2m}(2+\sqrt{3})$ and $2+\sqrt{3}\sim 3.7320508 \leq  3.7321$. Likewise, keeping $\ell=1$, $f_m(B)>2\sqrt{2m}$ if $\vol(B)<\sqrt{2m}(2-\sqrt{3})$ and $2-\sqrt{3}\sim 0.26794919 \geq 0.2679$. This establishes part 1 of the lemma. The parts 2 and 3 follow in the same fashion. \qed
\end{proof}

We are now ready to prove Theorem \ref{thm:hard}.

\begin{proof}[Proof of Theorem \ref{thm:hard}]
We give a reduction from AECP.  Suppose that $(H,\{a_1,\ldots,a_r\})$ is the input to an instance of AECP; we will describe how to construct a graph~$G$, where $\pw(G)$ and $\fvs(G)$ are both bounded by a function of $r$, together with an explicit $q_0 \in (0,1)$ such that $(G,q_0)$ is a yes-instance for \textsc{Modularity} if and only if $(H,\{a_1,\ldots,a_r\})$ is a yes-instance for AECP.

We may assume without loss of generality that our instance of AECP satisfies all of the conditions of Lemma \ref{lma:aecp-hard}.

We define a new graph~$G$, obtained from $H$ by adding the following (see Figure~\ref{fig:w1}):
\begin{itemize}
\item $\alpha$ new leaves adjacent to each anchor vertex $a_1,\ldots,a_r$,
\item $\beta$ isolated edges disjoint from~$G$, and
\item an arbitrary perfect matching on the anchor vertices $a_1,\ldots,a_r$,
\end{itemize} 
where the values of $\alpha$ and $\beta$ will be determined later.  
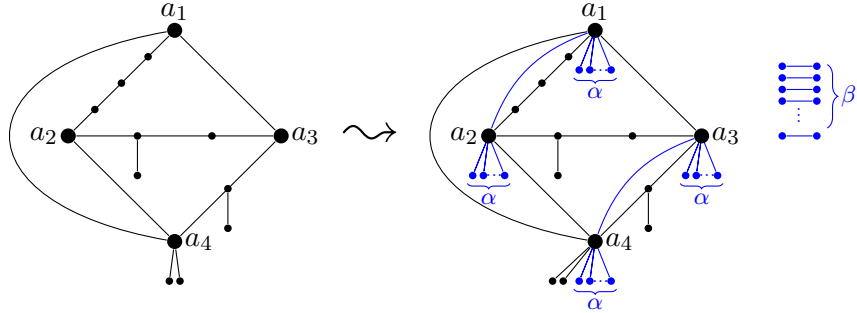
\begin{figure}[ht!]
  \begin{center}
    \begin{tikzpicture}[scale=0.7]
		\tikzstyle{vertex}=[circle,fill=black, draw=none, minimum size=4pt,inner sep=2pt]

%   \begin{scope}[shift={(-5,0)}]
		\node[vertex] (aB) at (0,0){};
		\node[vertex] (aL) at (-2,2){};
		\node[vertex] (aT) at (0,4){};
		\node[vertex] (aR) at (2,2){};

		\node[above] at (0,4) {$a_1$};
		\node[left] at (-2,2) {$a_2$};
		\node[right] at (0,0) {$a_4$};
		\node[right] at (2,2) {$a_3$};

		\draw (aT)--(aL)--(aB)--(aR)--(aT);
		\draw (aL)--(aR);
		\draw (aT) to [out=190,in=170, distance=4cm] (aB);

		\tikzstyle{vertex}=[circle,fill=black, draw=none, inner sep=1pt]
		\node[vertex] (vLR1) at (-0.7,2){};
		\node[vertex] (vLR2) at (0.7,2){};
		\node[vertex] (vBR) at (1,1){};
		
		\node[vertex] (vTL1) at (-1,3){};
		\node[vertex] (vTL2) at (-0.5,3.5){};
		\node[vertex] (vTL3) at (-1.5,2.5){};

		\node[vertex] (vBp1) at (-0.1,-0.75){};
		\node[vertex] (vBp2) at (0.1,-0.75){};
		\draw (vBp1)--(aB)--(vBp2);
		\node[vertex] (vBRp) at (1,0.25){};
	  	\draw (vBRp) -- (vBR);

		\node[vertex] (vLR1p) at (-0.7,1.25){};
		\draw (vLR1p) -- (vLR1);

		\node[] at (3.7,2) {\huge $\leadsto$};
 %  \end{scope}

    	%__________________________
      	%SECOND PICTURE
    	\begin{scope}[shift={(7.9,0)}]
  	\tikzstyle{vertex}=[circle,fill=black, draw=none, inner sep=1pt]

      \begin{scope}[shift={(2,2)}]
   	\node[vertex,blue] (ref2) at (0,0){};
		\node[vertex,blue] (a2) at (-0.3,-0.75){};
		\node[vertex,blue] (b2) at (-0.1,-0.75){};
		\node[vertex,blue] (c2) at (0.3,-0.75){};
		\node[blue] at (0.1,-0.75) {\tiny ...};
		
		\draw[blue] (ref2)--(a2)--(ref2)--(b2)--(ref2)--(c2);
		\draw [decorate, decoration={brace,amplitude=2pt,mirror},xshift=0pt,yshift=0pt, blue] (-0.4,-0.9) -- (0.4,-0.9) node [black,midway,below,xshift=0cm]  {\textcolor{blue}{\footnotesize $\alpha$}};
		\end{scope}

      \begin{scope}[shift={(-2,2)}]
		\node[vertex,blue] (ref3) at (0,0){};
		\node[vertex,blue] (a3) at (-0.3,-0.75){};
		\node[vertex,blue] (b3) at (-0.1,-0.75){};
		\node[vertex,blue] (c3) at (0.3,-0.75){};
		\node[blue] at (0.1,-0.75) {\tiny ...};
		
		\draw[blue] (ref3)--(a3)--(ref3)--(b3)--(ref3)--(c3);
		\draw [decorate, decoration={brace,amplitude=2pt,mirror},xshift=0pt,yshift=0pt, blue] (-0.4,-0.9) -- (0.4,-0.9) node [black,midway,below,xshift=0cm]  {\textcolor{blue}{\footnotesize $\alpha$}};
		\end{scope}

		\begin{scope}[shift={(0,4)}]
		\node[vertex,blue] (ref4) at (0,0){};
		\node[vertex,blue] (a4) at (-0.3,-0.75){};
		\node[vertex,blue] (b4) at (-0.1,-0.75){};
		\node[vertex,blue] (c4) at (0.3,-0.75){};
		\node[blue] at (0.1,-0.75) {\tiny ...};
		
		\draw[blue] (ref4)--(a4)--(ref4)--(b4)--(ref4)--(c4);
		\draw [decorate, decoration={brace,amplitude=2pt,mirror},xshift=0pt,yshift=0pt, blue] (-0.4,-0.9) -- (0.4,-0.9) node [black,midway,below,xshift=0cm]  {\textcolor{blue}{\footnotesize $\alpha$}};
		\end{scope}

		\node[vertex,blue] (ref) at (0,0){};
		\node[vertex,blue] (a) at (-0.3,-0.75){};
		\node[vertex,blue] (b) at (-0.1,-0.75){};
		\node[vertex,blue] (c) at (0.3,-0.75){};
		\node[blue] at (0.1,-0.75) {\tiny ...};
		
		\draw[blue] (ref)--(a)--(ref)--(b)--(ref)--(c);
		\draw [decorate, decoration={brace,amplitude=2pt,mirror},xshift=0pt,yshift=0pt, blue] (-0.4,-0.9) -- (0.4,-0.9) node [black,midway,below,xshift=0cm]  {\textcolor{blue}{\footnotesize $\alpha$}};

		\tikzstyle{vertex}=[circle,fill=black, draw=none, minimum size=4pt,inner sep=2pt]
		\node[vertex] (aB) at (0,0){};
		\node[vertex] (aL) at (-2,2){};
		\node[vertex] (aT) at (0,4){};
		\node[vertex] (aR) at (2,2){};

		\node[above] at (0,4) {$a_1$};
		\node[left] at (-2,2) {$a_2$};
		\node[right] at (0,0) {$a_4$};
		\node[right] at (2,2) {$a_3$};

		\draw (aT)--(aL)--(aB)--(aR)--(aT);
		\draw (aL)--(aR);
		\draw (aT) to [out=190,in=170, distance=4cm] (aB);

		\draw[blue] (aT) to [out=200,in=70] (aL); 
		\draw[blue] (aR) to [out=200,in=70] (aB);

		\tikzstyle{vertex}=[circle,fill=black, draw=none, inner sep=1pt]
		\node[vertex] (vLR1) at (-0.7,2){};
		\node[vertex] (vLR2) at (0.7,2){};
		\node[vertex] (vBR) at (1,1){};
		
		\node[vertex] (vTL1) at (-1,3){};
		\node[vertex] (vTL2) at (-0.5,3.5){};
		\node[vertex] (vTL3) at (-1.5,2.5){};

		\node[vertex] (vBp1) at (-0.6,-0.75){};
		\node[vertex] (vBp2) at (-0.8,-0.75){};
		\draw (vBp1)--(aB)--(vBp2);
		\node[vertex] (vBRp) at (1,0.25){};
	  	\draw (vBRp) -- (vBR);

		\node[vertex] (vLR1p) at (-0.7,1.25){};
		\draw (vLR1p) -- (vLR1);
		
		\node[rotate=90,blue] at (3.83,2.4) {\tiny$...$};

      \begin{scope}[shift={(3.515,3.32)}]
            \node[vertex,blue] (vL) at (0,0){};
            \node[vertex,blue] (vR) at (0.65,0){};
            \draw[blue] (vL)--(vR);
      \end{scope}

      \begin{scope}[shift={(3.515,3.1)}]
            \node[vertex,blue] (vL) at (0,0){};
            \node[vertex,blue] (vR) at (0.65,0){};
            \draw[blue] (vL)--(vR);
      \end{scope}

      \begin{scope}[shift={(3.515,2.88)}]
            \node[vertex,blue] (vL) at (0,0){};
            \node[vertex,blue] (vR) at (0.65,0){};
            \draw[blue] (vL)--(vR);
      \end{scope}

      \begin{scope}[shift={(3.515,2.66)}]
            \node[vertex,blue] (vL) at (0,0){};
            \node[vertex,blue] (vR) at (0.65,0){};
            \draw[blue] (vL)--(vR);
      \end{scope}

      \begin{scope}[shift={(3.515,2)}]
      		\node[vertex,blue] (vL) at (0,0){};
      		\node[vertex,blue] (vR) at (0.65,0){};
      		\draw[blue] (vL)--(vR);
		\end{scope}
		
      \draw [decorate, decoration={brace,amplitude=4pt},xshift=-4pt,yshift=0pt, blue] (4.5,3.35) -- (4.5,2.15) node [black,midway,xshift=0.05cm,right]  {\textcolor{blue}{\footnotesize $\beta$}};
      \end{scope}

    \end{tikzpicture}\vspace{-5mm}
  \end{center}
\caption{Possible input graph $H$ with anchors $a_1,a_2,a_3,a_4$ and the graph~$G$ constructed from it by adding $\alpha$ new leaves adjacent to each anchor, $\beta$ isolated edges and a perfect matching between the anchors.\vspace{-4mm}}
\label{fig:w1}
\end{figure}
\vspace{6pt}

{\Fc The idea of the construction is that the $\alpha$ edges help ensure that each anchor vertex is in a separate part of any modularity optimal partition and the $\beta$ edges allow us to get the numbers to work at the end of the proof.} Notice that, even with these modifications, $G \setminus \{a_1,\ldots,a_r\}$ is still a disjoint union of isolated vertices and paths with pendant edges; hence $\pw(G) \leq r+1$ and $\fvs(G) \leq r$.  We set $m=|E(G)|$ so $m=|E(H)| + \alpha r+\beta+r/2$.

Define our instance of \textsc{Modularity} to be $(G,q_0)$, where
\begin{equation*}
q_0 = 1-\frac{\beta}{m^2}-\frac{2\sqrt{2}(m-\beta)}{m^{3/2}}.
\end{equation*}

We now argue that $(G,q_0)$ is a yes-instance if and only if $(H,\{a_1,\ldots,a_r\})$ is a yes-instance for AECP. Recall that 
\begin{equation*}
\q(G)= 1-\min_\AA \sum_{A\in \AA}\bigg( \frac{\partial(A)}{2m}+ \frac{\vol(A)^2}{4m^2}\bigg).
\end{equation*}

\noindent
and that the partition $\AA$ which achieves the minimum in the expression above is exactly the modularity maximal $\AA$. In any modularity optimal partition, $\AA$, each isolated edge will form its own part: this follows from Facts~\ref{fact:singlevertex} %\footnote{Here, only need weaker 'pendant' vertex Fact of Brandes} 
and~\ref{fact:connected}. Write $V'$ for vertices of~$G$ without the vertices supporting the $\beta$ isolated edges, and let the minimisation be over $\AA'$ which are vertex partitions of $V'$. We then have \begin{equation*}
\q(G)= 1-\frac{\beta}{m^2}-\min_{\AA'} \sum_{A\in \AA'} \frac{\partial(A)}{2m}+ \frac{\vol(A)^2}{4m^2}.
\end{equation*}

\noindent
Rearranging, we see that \begin{eqnarray}
\notag 1-\frac{\beta}{m^2} - \q(G') 
&=&  \frac{m-\beta}{m}  \min_{\AA'} \sum_{A\in \AA'} \frac{\vol(A)}{2(m-\beta)} \left(\frac{\partial(A)}{\vol(A)}+ \frac{\vol(A)}{2m}\right)\\
\label{eq.exact} &=&  \frac{m-\beta}{m} \min_{\AA'}   \sum_{A\in \AA'} \frac{\vol(A)}{2(m-\beta)} f_m(A)\\
\label{eq.upperbound} &\geq &  \frac{m-\beta}{m} \min_{A\subseteq V'} f_m(A).
\end{eqnarray}
The last inequality holds because $\sum_A \vol(A)=2(m-\beta)$ and so~\eqref{eq.exact} is a weighted sum of the $f_m(A)$ with total weight one. This, together with the fact that no $A$ has zero volume, also implies that \eqref{eq.exact}$\geq$\eqref{eq.upperbound} with equality if and only if $f_m(A)=\min_{B \subset V'} f_m(B)$ for every $A\in \AA'$.

Note that, since $\AA'$ is the restriction of some modularity optimal partition $\AA$ to a connected component of~$G$, we may assume that, for all $A\in \AA'$, $G[A]$ is connected.  Moreover, if $v$ is a pendant vertex adjacent to $u$ then $u$ and $v$ are in the same part in $\AA'$; we call a partition with this last property (or, abusing notation, a set that would not violate this condition in a partition) `pendant-consistent'.

We now make the following claim, writing $s = |H|/r$ for the desired part size in our instance of AECP.
 
\begin{claim}\label{clm:atonce} 
Suppose that $\alpha > 32|E(H)|^2$ and that we have $\sqrt{2m} = s + \alpha + 1$.  Then:
\begin{enumerate}[label=\alph*)]
\item
for any connected, pendant-consistent set $B\subseteq V'$ we have $f_m(B)\geq 2\sqrt{2/m}$, and if $f_m(B) = 2\sqrt{2/m}$ then $B$ contains exactly one anchor and $\vol(B)=2\sqrt{2m}$;
\item if $(H,\{a_1,\ldots,a_r\})$ is a yes-instance, then there is a vertex partition $\AA'$ of $V'$ so that $f_m(A)=2\sqrt{2/m}$ for all $A \in \AA'$;
\item if there is a vertex partition $\AA'=\{A_1, \ldots, A_r\}$ of $V'$ so that for all $A_i\in \AA$, $f_m(A_i)=2\sqrt{2/m}$, $A$ is pendant-consistent and $G[A]$ is connected for all $A \in \AA'$, then $(H,\{a_1,\ldots,a_r\})$ is a yes-instance.
\end{enumerate}
\end{claim}

We defer the proof of Claim~\ref{clm:atonce}. For now we assume that Claim~\ref{clm:atonce} holds and that we have $\alpha > 32|E(H)|^2$ and $\sqrt{2m} = s + \alpha + 1$ and prove the theorem holds under these assumptions. 
%\noindent
By Claim~\ref{clm:atonce}(a) and line \eqref{eq.exact}, we %always 
have $$q^*(G) \leq q_0 = 1-\frac{\beta}{m^2}-\frac{2\sqrt{2}(m-\beta)}{m^{3/2}}.$$
Hence in particular $(G,q_0)$ is a yes-instance if and only if there is a partition $\AA'$ of $V'$ such that $\forall A\in \AA'$ $f_m(A)=2\sqrt{2/m}$.

Claim~\ref{clm:atonce}(b), together with line \eqref{eq.exact}, implies that, if $(H,\{a_1,\ldots,a_r\})$ is a yes-instance, then so is $(G,q_0)$.  Converesly, if $(G,q_0)$ is a yes-instance, it follows from Claim~\ref{clm:atonce}(c), 
that $(H,\{a_1,\ldots,a_r\})$ is a yes-instance.

It remains only to show the claim holds and we can choose suitable values of $\alpha$ and $\beta$ to ensure $\alpha > 32|E(H)|^2$ and $\sqrt{2m} = s + \alpha + 1$. Set $\alpha$ to be the least integer such that 
\begin{equation}\label{eq.defalpha}
\alpha\geq 32|E(H)|^2, \;\;\; (\alpha+s+1)^2>2|E(H)|+2\alpha r+r \;\;\; \mbox{and} \;\;\; \alpha=s+1 \;(\mbox{mod } 2). 
\end{equation}
Recall that $r$ is even. This, along with our parity constraint between $\alpha$ and $s$, implies that $(\alpha+s+1)^2-r$ is even. Thus we can choose $\beta$ to be
\begin{equation}\label{eq.defbeta}
\beta = \frac{1}{2}\left( (\alpha+s+1)^2-r \right)- |E(H)|-\alpha r;
\end{equation}
note $\beta$ is positive because we set $\alpha$ so that $(\alpha+s+1)^2>2|E(H)|+2\alpha r+r$. Finally, observe that we do have $\sqrt{2m}=s+\alpha+1$ because, by the chosen value of $\beta$,
\begin{equation*} m= |E(H)|+\alpha r+\beta+r/2 = (s+\alpha+1)^2/2.\end{equation*}

This concludes the proof of the theorem except that we must still establish Claim~\ref{clm:atonce}.

\begin{proof}[Proof of Claim \ref{clm:atonce}(a):] We begin by showing that our two assumptions $\alpha>32|E(G)|^2$ and $\alpha+s+1=\sqrt{2m}$ imply that $\alpha > 0.969\sqrt{2m}$. Recall that $H$ without pendant edges is a subdivision of a cubic graph and so the average degree in $H$ is at least two. Thus $|E(H)|\geq |H|$. Also $r\geq 4$, so $|H|\geq 4s \geq s+1$ so $|E(H)|\geq s+1$. By assumption $\alpha > 32 |E(H)|^2 \geq 32(s+1)$. But also by assumption $\alpha+s+1=\sqrt{2m}$ and so $\alpha\geq 32/33\sqrt{2m} > 0.969\sqrt{2m}$.

We now show that if $f_m(B)\leq 2\sqrt{2/m}$ then $B$ must contain exactly one anchor. First suppose $B$ contains no anchors: then $B$ does not contain nor is $B$ incident to any of the $\alpha$ extra edges added to anchors nor the $r/2$ extra edges in the perfect matching between anchors. Hence the volume of $B$ in~$G$ is at most what it was in $H$, i.e. $\vol_{G}(B)\leq \vol_H(B) \leq 2|E(H)|$. Also note that as $G[V']$ is connected, $\partial(B)\geq 1$, hence by Lemma~\ref{lem.maxf4} it is enough to show that $\vol(B)<0.2679\sqrt{2m}$ and this will show that for $B$ with no anchors, $f_m(B)>2\sqrt{2/m}$. Clearly $m\geq \alpha$ and by assumption $\alpha>32|E(H)|^2$. Hence for $B$ with no anchors:
\[0.2679\sqrt{2m}\geq 0.2679\sqrt{64 |E(H)|^2} = 2.1432 \; E(H) > \vol(B)\]
and so if $f_m(B)\leq 2\sqrt{2/m}$ then $B$ must contain at least one anchor.

If $B$ contains at least two anchors then there are two options: $B=V'$ and $B\subsetneq V'$. We rule out $B=V'$ and $f_m(B)\leq 2\sqrt{2/m}$ first. Note that $\vol(V')=2|E(H)|+2\alpha r+r$. But $r\geq 4$ and by earlier in the proof $\alpha>0.969\sqrt{2m}$. Hence $\vol(V')\geq 8\alpha > 7.752\sqrt{2m}$ and so by Lemma~\ref{lem.maxf4} we get that $f_m(V')>2\sqrt{2/m}$. Thus $B\neq V'$.

Now we show that for $B\subsetneq V'$ with at least two anchors in $B$ we have $f_m(B)>2\sqrt{2/m}$. In the case $B\neq V'$ because $G[V']$ is connected $\partial(B)\geq 1$. If $B$ has at least two anchors then $\vol(B)\geq 4\alpha>3.876\sqrt{2m}$. Therefore for $B\subsetneq V'$ with at least two anchors in $B$, $\partial(B)\geq 1$ and $\vol(B) > 3.878\sqrt{2m}$ hence $f_m(B)>2\sqrt{2/m}$ by Lemma~\ref{lem.maxf4}.

Thus to ensure $f_m(B)\leq 2\sqrt{2/m}$ we must have exactly one anchor in $B$. In particular we can now assume that $B$ contains exactly one anchor. Let graph $G'$ be~$G$ without the added perfect matching between anchors at the end of the construction of~$G$ from $H$. Now $G'[B]$ is connected, $B$ has exactly one anchor and after stripping pendant vertices that anchor has degree $3$ in $G'$ so we have $\partial_{G'}(B)\geq 3$. And because $B$ is pendant-consistent $\partial_{G'}(B)=3$, after re-adding the perfect matching between anchors $\partial_{G}(B)=4$.

But now, because $\partial_G(B)=4$, by Lemma~\ref{lem.maxf4} we have that $f_m(B)\geq 2\sqrt{2/m}$. Also by Lemma~\ref{lem.maxf4} to get equality $f_m(B)=2\sqrt{2/m}$ we must have $\vol(B)=2\sqrt{2m}$ which establishes the last part of the claim.
\hfill $\square$ \textit{Claim \ref{clm:atonce}(a)}
\end{proof}

\begin{proof}[Proof of Claim \ref{clm:atonce}(b):]
Suppose $(H,\{a_1,\ldots,a_r\})$ is a yes-instance.  We prove there exists a vertex partition $\AA'$ of $V'$ such that, for all $A\in \AA'$, $f_m(A)=2\sqrt{2/m}$. By assumption, there is a connected equipartition $\BB=\{B_1,\ldots,B_r\}$ of $V(H)$ such that $a_i \in B_i$ for each $i$. In the construction of the graph~$G$ from $H$ we added $\alpha$ pendant vertices, say $u_1^i, \ldots, u_\alpha^i$, to each anchor~$a_i$. Define $\AA'=\{ B_i \cup \{ u_1^i, \ldots, u_\alpha^i\} \; : \; B_i \in \BB\}$. Observe that that $\AA'$ is a vertex partition of $V'$ as the set $V'$ consists exactly of $V(H)$ together with the extra~$\alpha r$ vertices added with the pendant edges on each anchor. It now remains to prove that $f_m(A_i)=2\sqrt{2/m}$ for each $i$.

Consider $G'$, the graph formed from~$G$ by removing the arbitary perfect matching added in the last step of the construction of~$G$ from $H$. Recall the graph $G'$ is the subdivision of a 3-regular graph with the anchors as the branch vertices. Fix $i$ and note that $H[B_i]$ connected implies that $G'[A_i]$ is connected. But as $G'[A_i]$ is connected, contains exactly one anchor and contains every vertex pendant to a vertex in $A_i$ it must be the case that $\partial_{G'}(A_i)=3$. Now re-add the perfect matching and we get that $\partial_G(A_i)=4$.

It suffices now to show that $\vol(A_i)=2\sqrt{2m}$. To see this, first observe that $G[A_i]$ is a tree, and so $|E_{G}(A_i)|=|A_i|-1$. But by the construction $|A_i|=|B_i|+\alpha=s+\alpha$. Recall the volume of a vertex set is twice the number of internal edges plus the number of edges between the set and the rest of the graph. Thus, because $\partial_{G}(A_i)=4$, we get 
$$\vol(A_i)=2(s+\alpha-1)+4 = 2(s+\alpha+1) = 2\sqrt{2m},$$
which establishes the claim.
\hfill $\square$ \textit{Claim \ref{clm:atonce}(b)}
\end{proof}

\begin{proof}[Proof of Claim \ref{clm:atonce}(c):] Suppose there exists a vertex partition $\AA'=\{A'_1, \ldots, A'_r\}$ of $V'$ such that, for all $A'_i\in \AA'$, $f_m(A'_i)=2\sqrt{2/m}$, $A'_i$ is pendant-consistent, and $G[A_i']$ is connected. By Claim~\ref{clm:atonce}(a) we may also assume that for all $A_i'\in\AA'$ we have $\vol(A_i')=2\sqrt{2m}$. We will show this implies that $(H,\{a_1,\ldots,a_r\})$ is a yes-instance.

Fix some $i$. The induced subgraph $G[A'_i]$ is connected and contains exactly one anchor, say $a_i$, so we can remove the perfect matching between the anchors and $G'[A'_i]$ is still connected. Let $B_i$ be the vertex set obtained from $A_i'$ by removing the $\alpha$ added leaves pendant on the anchor $a_i$. Then $B_i\subseteq V(H)$ and $H[B_i]$ is connected.

It remains only to show that $|B_i|$ is exactly $s=|H|/r$. Since $G(A_i')$ is a tree with volume $2\sqrt{2/m}$ and $\partial_{G}(A_i')=4$, $\vol(A_i')=2(|A_i|-1)+4=2|A_i|+2$. But $|A_i|=|B_i|+\alpha$ and so 
$$|B_i|=|A_i|-\alpha = \vol(A_i')/2-1-\alpha;$$
by design this is precisely $\vol(A_i')/2-1-\alpha=s$ and so we are done.\\
\text{} \hfill $\square$ \textit{Claim \ref{clm:atonce}(c)}
\end{proof}

This completes the proof.\qed
\end{proof}

\section{Conclusions and Open Problems}

We have shown that \textsc{Modularity} belongs to $\FPT$ when parameterised by the vertex cover number of the input graph, and that the problem is solvable in polynomial time on input graphs whose treewidth or max leaf number is bounded by some fixed constant; we also showed that there is an FPT algorithm, parameterised by treewidth, which computes any constant-factor approximation to the maximum modularity.  In contrast with the positive approximation result, we demonstrated that the problem is unlikely to admit an exact FPT algorithm when the treewidth is taken to be the parameter, as it is $\W[1]$-hard even when parameterised simultaneously by the pathwidth and size of a minimum feedback vertex set for the input graph.  

We conjecture that our XP algorithm parameterised by max leaf number is not optimal, and that \textsc{Modularity} in fact belongs to $\FPT$ with respect to this parameterisation.  Another open question arising from our work is whether the problem belongs to $\FPT$ with respect to other parameters for which this is not ruled out by our hardness result, including treedepth, modular width and neighbourhood diversity.

It is also natural to ask whether our approximation result can be extended to larger classes of graphs, for example those of bounded cliquewidth or bounded expansion.  Moreover, when considering treewidth as the parameter, it would be interesting to investigate the existence or otherwise of an $\epsilon$-approximation in time $f(\tw,\epsilon) n^{\mathcal{O}(1)}$.

\paragraph*{Acknowldegements}
The authors are grateful to Jessica Enright for some helpful initial discussions about the topic.

%\bibliography{\modularity_refs}   % name your BibTeX data base

%% Non-BibTeX users please use

\end{document}